\documentclass[twoside]{article}
\usepackage[a4paper]{geometry}
\usepackage{palatino}
\usepackage[utf8]{inputenc}
\usepackage[T1]{fontenc} 
\usepackage{RR}
\usepackage{hyperref}
\usepackage{color}
\usepackage{amsfonts}
\usepackage{amsmath}
\usepackage{amsthm}
\usepackage{amssymb}
\usepackage{mathrsfs}
\usepackage{mathalfa}
\usepackage{empheq}

\newcommand{\dudx}[2]
{\displaystyle{\frac{\partial #1}{\partial #2}}}
\newcommand{\DtPerp}[1]
{\displaystyle{\frac{D^\perp}{Dt}#1}}
\newcommand{\gradPar}[1]
{\displaystyle{\nabla_{/\!/} #1}}
\newcommand{\vc}[1]{\mbox{$\boldsymbol{#1}$}} 
\definecolor{grispale}{gray}{0.75}

\newcommand{\grandR}{I\!\!R}

\newcommand{\vecx}{{\rm \bf e_x}}
\newcommand{\vecy}{{\rm \bf e_y}}
\newcommand{\vecz}{{\rm \bf e_z}}

\RRNo{8715}
\RRdate{April 2015}
\RRauthor{
Herv\'e Guillard\thanks{herve.guillard@inria.fr}}
\authorhead{H.Guillard}
\RRtitle{Th\'eorie math\'ematique des mod\`eles de MHD r\'eduite pour les plasmas de fusion}
\RRetitle{The mathematical theory of reduced MHD models for fusion plasmas}
\titlehead{Mathematical theory of reduced MHD models}
\RRresume{L'établissement de modèles de MHD réduite est formulé comme un exemple de la 
théorie générale des limites singulières des systèmes hyperboliques. Cette formulation permet 
d'utiliser les résultats généraux de cette théorie et de prouver rigoureusement que les modèles 
de MHD réduite sont une approximation valide du modèle complet. En particulier, 
la convergence des solutions du modèle complet vers les solutions d'un système réduit est 
démontrée. 
}
\RRabstract{The derivation of reduced MHD models for fusion plasma is here 
formulated as a special instance of the general theory of singular limit of 
hyperbolic system of PDEs with large operator. This formulation allows to use the 
general results of this theory and to prove rigorously that reduced MHD models are 
valid approximations of the full MHD equations. In particular, it is proven that 
the solutions of the full MHD system converge to the solutions of an appropriate 
reduced model.}
\RRmotcle{Analyse asymptotique, systèmes hyperboliques, limite singulière, MHD, Plasmas de fusion, 
Tokamaks}

\RRkeyword{Asymptotic analysis, hyperbolic systems, singular limit, MHD, Fusion plasma, 
Tokamaks}
 \RRprojet{Castor}  
\RCSophia 

\begin{document}
\makeRR   
\section{Introduction} \label{Intro}
Magnetohydrodynamics (MHD) is a macroscopic theory describing electrically conducting 
fluids. It addresses laboratory as well as astrophysical plasmas and 
therefore is extensively used in very different contexts. One of these contexts concerns 
the study of  
fusion plasmas in tokamak machines. A tokamak is a toroidal device in which hydrogen 
isotopes in the form of a plasma reaching a temperature of the order of the 
hundred of millions of Kelvins 
is confined thanks to a very strong applied magnetic field. Tokamaks are used to study 
controlled fusion and are considered as one of the most promising concepts to 
produce fusion energy in the near future. However a hot plasma as the one 
present in a tokamak 
is subject to a very large number of instabilities that can lead to the 
end of the existence of the plasma. An important goal of MHD 
studies in tokamaks is therefore to determine the stability domain that constraints 
the operational range of the machines. 
A secondary goal of these studies is to evaluate the consequences of these possible 
instabilities in term of heat loads and stresses on the plasma facing components of the 
machines. 
Numerical simulations using the MHD models are therefore of uttermost importance in this field 
and therefore the design of MHD models and of models beyond the standard one 
(e.g incorporating two-fluid or kinetic effects) is the subject of an 
intense activity. \\
However, the MHD model is a very 
complex one : it contains 8 independent variables, three velocity components, 
three components of the magnetic field, density and pressure. Although the system is 
hyperbolic, it is not strictly hyperbolic leading to the existence of possible resonance 
between waves of different types and moreover the MHD system has the additional complexity 
of being endowed with an involution. An
involution in the sense of conservation law systems is an additional equation that
if satisfied at $t=0$ is satisfied for all $t>0$ \cite{DafermosInvolution}. 
For all these reasons, approximations and simplified models have been designed 
both for theoretical studies as well as numerical ones. In the field of fusion plasmas, 
these models are denoted as {\em reduced MHD models}\footnote{while the standard MHD model is 
by contrast designated as the {\em full} MHD model}. These models initially 
proposed in the 70' \cite{Strauss76}) have been 
progressively refined to 
include more and more physical effects and corrections \cite{Strauss77, Carreras-1981, 
Schmalz198114, Izzo-etal-1983}. In particular, some earlier models conserve a non-standard energy 
and in some modeling works, special attention have been paid to insure the conservation of the usual 
energy e.g \cite{Strauss1997, Drake_Antonsen1984, Kruger_etal98} 
(see also \cite{franck_etal2015}). At present the literature  on the physics of fusion plasma 
concerned by reduced MHD models is huge and contains several hundred of references. 
From a numerical point of view, several well-known numerical codes (e.g 
\cite{0741-3335-54-12-124047}, \cite{Czarny20087423} ) used  routinely for fusion 
plasma studies are based on these reduced models. Actually, while there is a definite tendency 
in the fusion plasma community to use full MHD models e.g \cite{NimrodFE2004}, 
\cite{M3DC1-51647842}, \cite{Haverkort2013},
a large majority of non-linear simulations of 
tokamak plasmas have been and still are conducted with these approximations. \\

Until recently, reduced MHD models have not attracted a lot of interest in the 
mathematical or numerical analysis literature. One can cite \cite{MZA:8489301} and 
\cite{franck_etal2015} that have shown that these models can be interpreted as some special 
case of ``Galerkin'' methods where the velocity and magnetic fields are constrained to 
belong to some lower dimensional space. This interpretation is also implicit in the 
design of the M3D-C1 code \cite{M3DC1-51647842} where instead of the usual projection 
on the coordinate system axis, the equations governing the scalar components of the vector fields 
are obtained by special projections that allow to recover reduced models. \\

In this work, we adopt the different point of view of asymptotic analysis and show that 
reduced MHD models can be understood as a special instance of the general theory of 
singular limit of hyperbolic system of PDEs with large constant operators. This formulation allows to use the 
general results of this theory and to prove rigorously the validity of these 
approximations of the MHD equations. In particular, it is proven here, we believe for the 
first time, that the solutions of the full MHD system converge to the solutions of an 
appropriate reduced model.

This paper is organized as follows : First, we recall the general theory of 
singular limits of quasi-linear hyperbolic system with a large parameter. 
In the third section, we show how this 
general framework can be used to analyze reduced MHD models. Finally, we conclude by some 
remarks on possible extensions of the present work. 

\section{Singular limit of hyperbolic PDEs}\label{sec:singularLimit}
\subsection{General framework}\label{general framework}
In this section, we are concerned with the behavior when $\varepsilon \rightarrow 0$ of the solutions of hyperbolic system of 
PDEs of the following form : 
\begin{equation}\label{eq:Hsystem}
\left\{\begin{array}{l}
A_0(\vc{W},\varepsilon) \partial_t \vc{W} + \sum_j A_j (\vc{W},\varepsilon)\partial_{x_j} 
\vc{W}+  \dfrac{1}{\varepsilon} \sum_j
C_j\partial_{x_j} \vc{W} =0  \\
\vc{W}(0,\vc{x},\varepsilon) = \vc{W}_0(\vc{x},\varepsilon) 
\end{array}\right.
\end{equation} 
Here $\vc{W} \in \mathcal{S} \subset \grandR^N$ is a vector function 
depending of $(t,x_j;j=1,\cdots,d)$ where $d$ is the space dimension while the $A_0, A_j, C_j$ are square $N \times N$ matrices. Due to the presence of the large coefficient 
$1/\varepsilon$ multiplying the operator $\sum_j C_j\partial_{x_j}(.)$,  we may expect 
 the velocity of some waves present in (\ref{eq:Hsystem}) to become infinite and 
 therefore, for a solution to exist on a ${\cal{O}}(1)$ time scale, 
 it has to be close in some sense to the kernel 
 $ K=\{ \vc{W} \in \grandR^N s.t \sum_j C_j\partial_{x_j} \vc{W}  = 0 \}$ 
 of the large operator. The limit system 
obtained from (\ref{eq:Hsystem}) is therefore a {\em singular} limit since the constraint 
$\vc{W} \in K $ may change the hyperbolic nature of the system (\ref{eq:Hsystem}). 
A prototypical example of this behavior is given by the incompressible limit of the 
{\em hyperbolic} equations governing compressible Euler flows where the propagation at 
infinite speed of the acoustic waves gives rise to an {\em elliptic} equation on the pressure coming 
from the global constraint $\nabla\cdot \vc{u} =0$. \\

The nature of the singular limit depends on the initial data. Using the terminology of 
Schochet \cite{S_Schochet_1994} the limit is called ``slow'' if the initial data makes 
the first time derivatives at time $t=0$ stay bounded as $\varepsilon \rightarrow 0$. 
The term ``well-prepared initial data'' is also used to qualify this situation. 
In this case, under appropriate assumptions, the solutions exist for a time $T$ independent of $\varepsilon$ and converge
to the solutions of a limit system when $\varepsilon \rightarrow 0$.\\

In the opposite case, denoted as a ``fast'' singular limit, $\partial_t \vc{W}$ is not 
${\cal{O}}(1)$ at time zero and fast oscillations developing on a $1/\varepsilon$ time 
scale can persist on the long time scale. Solutions of fast singular limit cannot converge 
as $\varepsilon \rightarrow 0$ in the usual sense since the time derivative of the solution 
is of order $1/\varepsilon$. In this case, convergence means the existence of an 
``averaged'' limit profile ${\cal{W}}(t,\tau,\vc{x})$ such that 
$\vc{W}(t,\vc{x},\varepsilon) - {\cal{W}}(t,t/\varepsilon,\vc{x}) \rightarrow 0$ with 
$\varepsilon$. The question of the existence of fast singular limit is in particular examined 
in \cite{S_Schochet_1994}. A review article summarizing results on this subject 
with a special emphasis on the low Mach number limit is \cite{Alazard_2008}. \\

In this work, we will be mainly concerned by 
the slow case. Even in this case, the existence for a time independent of $\varepsilon$ and 
the convergence of the solutions to the solutions of a 
limit system may require additional assumptions on the structure of (\ref{eq:Hsystem}). 
Beginning with the earlier works in the 80' of Klainerman and Majda
\cite{S_Klainerman_A_Majda_1981,S_Klainerman_A_Majda_1982, Majda_84} and those of 
Kreiss and his co-workers \cite{HO_Kreiss_1980,G_Browning_HO_Kreiss_1982}, 
these questions have 
been examined in several works \cite{S_Schochet_1986,S_Schochet_1988} with the main 
objective to justify the passage to the incompressible limit in low Mach number compressible 
flows. Several extensions of these works for viscous flows or general 
hyperbolic-parabolic systems are also available. Again we can refer to \cite{Alazard_2008} 
for a review.  

The following theorem (see \cite{Majda_84}, chapter 2)  summarizes the 
main results of 
these works in a form suitable for our purposes : \\
\newtheorem{T1}{Theorem 1}
\begin{T1}
 
Assume that :
\begin{enumerate}
\item {Conditions on the initial data : } 
$W_0(\vc{x},\varepsilon) = W_0^0(\vc{x})+ \varepsilon W_0^1(\vc{x},\varepsilon)$
     \begin{enumerate}
      \item $\vc{W}_0^0(\vc{x})$ and $\vc{W}_0^1(\vc{x},\varepsilon)$ are in $H^s$
      \item $\sum_j C_j \partial_j \vc{W}_0^0 = 0$
      \item $||\vc{W}_0^1(\vc{x},\varepsilon)||_s \le $ Cte
     \end{enumerate}
 \item {Structure of the system} 
       \begin{enumerate}
          \item The matrices $A_0,A_j$ and $C_j$ are symmetric 
          \item $A_0$ is positive definite at least in a neighborhood of the initial data
          \item $A_0$ and $A_j$ are $C^s$ continuous for some $s \ge [n/2]+2$, where $n$
          is the number of spatial dimensions
          \item The $C_j$ are constant matrices
          \item The matrix $A_0(\vc{W},\varepsilon) = A_0 (\varepsilon \vc{W})$
       \end{enumerate}

\end{enumerate}
then the solution $\vc{W}(t,\vc{x},\varepsilon)$ of system (\ref{eq:Hsystem}) with the 
initial data satisfying condition 1 is unique and exists for a time $T$ independent of
$\varepsilon$. In addition the solutions $\vc{W}(t,\vc{x},\varepsilon)$ satisfy :
\begin{displaymath}
 ||\vc{W}(t,\vc{x},\varepsilon) - \vc{W}^0(t,\vc{x})||_{s-1} \le C\varepsilon 
 \mathrm{~for~} t \in [0,T]
\end{displaymath}
where $\vc{W}^0(t,\vc{x})$ is the solution of the reduced system : 
\begin{equation}\label{eq:limitsystemG}
\left\{\begin{array}{l}
A_0(0) \partial_t \vc{W^0} + \sum_j A_j (\vc{W}^0, 0)\partial_{x_j} \vc{W^0}+  
\sum_j C_j\partial_{x_j} \vc{W^1} =0  \\
\sum_j C_j\partial_{x_j} \vc{W^0}=0\\
\vc{W^0}(0,\vc{x}) = \vc{W}_0^0(\vc{x})
\end{array}\right.
\end{equation} 
\end{T1}
\begin{proof}: The proof of this result can be found in \cite{Majda_84}, chapter 2, 
Theorems 2.3 and 2.4. 
We do not repeat this proof 
here but briefly comment on some of their aspects\nobreak : The assumptions 2.(a) et  2.(b)  
simply means that system (\ref{eq:Hsystem}) is a quasi-linear symmetric 
hyperbolic system in the sense of  Friedrichs. The uniqueness and existence of solution on a 
finite time $T >0 $ can 
then be established by classical iteration techniques relying on energy estimates 
(see for instance \cite{Lax_73} or \cite{Majda_84}). 
However the presence of the large 
coefficient $1/\varepsilon$ could possibly make this time of existence 
$\varepsilon$-dependent and shrinking to 0 with $\varepsilon$. 
Assumption 2.(d) ensures that this will not be the case since the 
matrices $C_j$ being constant,  the 
large terms will not contribute to the energy estimates. \\

\noindent The assumption 2.(e) $A_0 = A_0 (\varepsilon \vc{W})$ allows to bound its time derivative
independently of $\varepsilon$ : Since we have 
 $$
  \partial_t A_0(\varepsilon \vc{W}) = \dfrac{DA_0}{D\vc{W}} ~\varepsilon \partial_t {\vc{W}}
  = -\dfrac{DA_0}{D\vc{W}} ~ \varepsilon A_0^{-1} [A_j \partial_j {\vc{W}}+ 
 \dfrac{1}{\varepsilon}  \partial_j C_j {\vc{W}}]$$
 The $\varepsilon$ and $1/\varepsilon$ terms balance together and give an estimate independent 
 of $\varepsilon$.\\
 

The assumptions 1.(b) and 1.(c) means that the initial condition is sufficiently close to the 
kernel of the large operator to ensure that the time derivative 
$\partial_t \vc{W}(0,\vc{x},\varepsilon)$ is bounded in $H^{s-1}$ independently of 
$\varepsilon$. This condition implies that the initial data are ``well-prepared'' and 
will not generate fast oscillations on a $1/\varepsilon$ time scale. 
\end{proof}
\subsection{Reduced limit system}\label{sec:reduced}
Even if (\ref{eq:limitsystemG}) provides a complete description of the behavior of the 
solutions of the original system as $\varepsilon$ tends to $0$, the limit system 
contains as many unknowns as the original one. Actually, one may even consider that it contains 
more unknowns as the first order correction $\vc{W}^1$ have also to be computed. In practice, this largely 
depends on the specific system considered as some lines of the matrices $A_j (\vc{W}^0, 0)$ 
may be identically zero and/or the evaluation of some terms of the first order correction
can be completely obvious. However, it can be interesting 
to derive from (\ref{eq:limitsystemG}) a ``reduced'' set of equations containing less unknowns 
by eliminating the first order correction
$\vc{W}^1$. A particularly pleasant framework to 
construct such a reduced system is the following : \\

Assume that the kernel 
$K=\{ \vc{W} \in \grandR^N s.t \sum_j C_j\partial_{x_j} \vc{W}  = 0 \}$ 
have dimension $n<N$ and can be parametrized 
by a linear operator with constant coefficients such that :
\begin{displaymath}
 \forall \vc{W} \in K \subset \grandR^N, ~~~~~\exists \vc{\omega} \in S \subset \grandR^n, 
 ~~~~\vc{W} = {\cal{M}}(\vc{\omega}) 
\end{displaymath}

Since the operator $\mathbb{L}=\sum_j C_j\partial_{x_j} $ have constant coefficients, 
${\cal{M}}(\vc{\omega})$ is also a differential operator of order 1 with constant coefficients that can be written : 
\begin{equation}\label{eq:maxwellian}
  {\cal{M}}(\vc{\omega}) = (\sum_{j=1}^d P_j \partial_{x_j} + P_0) \vc{\omega}
\end{equation}
where the matrices $\{P_j;j=0,d\}$ are rectangular $N\times n$ constant matrices. 
Then consider the adjoint operator ${\cal{M}}^*$ from $\grandR^N$ to $\grandR^n$ satisfying 
\begin{displaymath}
  ({\cal{M}}(\vc{\omega}),\vc{W}) = (\vc{\omega},{\cal{M}}^*\vc{W})
\end{displaymath}
The operator ${\cal{M}}^*$ is an  ``annhilator'' for the $\mathbb{L}$ operator in the sense that 
\begin{displaymath}
  {\cal{M}}^*\mathbb{L} = 0
\end{displaymath}
Indeed we have : 
\begin{displaymath}
  ({\cal{M}}^*\mathbb{L}\vc{W}, \vc{\omega})= (\mathbb{L}\vc{W}, {\cal{M}}\vc{\omega})=
  -(\vc{W}, \mathbb{L}{\cal{M}}\vc{\omega})= 0
\end{displaymath}
since the $C_j$ being symmetric matrices, $\mathbb{L}$ is a skew-symmetric operator. \\
From (\ref{eq:maxwellian}) ${\cal{M}}^*$ has the explicit expression : 
\begin{equation}\label{eq:Amaxwellian}
  {\cal{M}}^*(\vc{W}) = - \sum_{j=1}^d P_j^t \partial_{x_j}\vc{W}  + P_0^t \vc{W}
\end{equation}
where $P_j^t;j=0,\cdots,d$ are rectangular $n\times N$ matrices, transposes of the 
$P_j$. \\

With the operators ${\cal{M}}$ and ${\cal{A}}={\cal{M}}^*$ at hand, a reduced system of 
equations can be obtained by   
left multiplying (\ref{eq:limitsystemG}) by the annhilator ${\cal{A}}$ for functions  
$\vc{W}={\cal{M}}(\vc{\omega})$. In this operation, 
the first-order correction 
$\sum_j C_j\partial_{x_j} \vc{W^1}$ vanishes and we obtain with 
$\vc{W} = {\cal{M}}(\vc{\omega})$ : 
\begin{equation}\label{eq:limitsystemForOmega}
\left\{\begin{array}{l}
{\cal{A}} A_0(0) {\cal{M}} ~\partial_t \vc{\omega} + 
\sum_j {\cal{A}} A_j ({\cal{M}}(\vc{\omega}), 0) {\cal{M}}~\partial_{x_j}  \vc{\omega} = 0 \\
\vc{\omega}(0,\vc{x}) = \vc{\omega}^0(\vc{x})
\end{array}\right.
\end{equation} 
that is an autonomous system for the reduced variable $\vc{\omega} \in \grandR^n$. Note that 
to obtain (\ref{eq:limitsystemForOmega}), we have used the fact that ${\cal{M}}$ being 
a linear differential operator defined by constant matrices $P_j$, it commutes with the 
time and spatial derivatives. \\
Note also that spatial derivatives are ``hidden'' in the 
definition of the operators ${\cal{A}}$ and ${\cal{M}}$. Therefore in contrast with the equations 
(\ref{eq:limitsystemG}) that is a first-order differential system,
~~(\ref{eq:limitsystemForOmega}) 
defines a {\em{third-order}} differential system of equations (see section \ref{sec:slowlimit} 
for the  concrete example of reduced MHD system). 
The choice of using (\ref{eq:limitsystemForOmega}) instead of 
(\ref{eq:limitsystemG}) as a basis for a numerical method is therefore problem dependent 
and in practice (\ref{eq:limitsystemForOmega}) can be more difficult to approximate than 
the original limit system. 
\section{Application to reduced MHD}\label{sec:Application to reduced MHD}

\subsection{The ideal MHD system}
We now proceed to show how this general framework can be applied to the MHD equations and 
begin to recall some basic facts about this system. In the sequel, we will make the assumption 
that the flow is barotropic, that is the pressure is only a function of the density. 
This assumptions includes isentropic as well as isothermal flows. \\
The ideal MHD system can be written under many different forms. Since the general theory 
we have described make use of the symmetry of the jacobian matrices, we use here a symmetric form 
of the system : 
 \begin{displaymath}\label{eq:MHDsym}
 \refstepcounter{equation}
\begin{array}{lll}
~~\rho \dfrac{D}{Dt} \vc{u} + \nabla (p + \vc{B}^2/2) - (\vc{B}.\nabla) \vc{B} &= 0 
& (\theequation.1)\\
\\
~~~~~\dfrac{D}{Dt} \vc{B} - (\vc{B}.\nabla) \vc{u}  + \vc{B} \nabla.\vc{u} &= 0
& (\theequation.2)\\
\\
\dfrac{1}{\gamma p} \dfrac{D}{Dt}p +  \nabla.\vc{u} &= 0 & (\theequation.3)
\end{array}
\end{displaymath}
In these equation, $\vc{u}$ is the velocity, $\vc{B}$ the 
magnetic field and $p$ is the pressure. The density $\rho$ is related to the 
pressure by a state law $\rho=\rho(p)$, for instance the perfect gas
state law that writes $\rho = A (p/s)^{1/\gamma}$ where $A$ and $\gamma$ are 
constant and $s$ is the (here constant) entropy. The notation $D./{Dt}$ stands for the material derivative that is 
defined by $D\cdot/{Dt} = \partial_t \cdot + (\vc{u}.\nabla)\cdot$. \\
To system (\ref{eq:MHDsym}) one must add the involution : 
\begin{displaymath}\label{eq:divB}
\refstepcounter{equation}
  \nabla\cdot\vc{B} = 0 ~~~~~~~~~~~~~~~~~~~~~~~~~~~~~~~(\theequation)
\end{displaymath}
and it is easily checked that if (\ref{eq:divB}) is verified at $t=0$, 
it is verified for all ~~$t>0$.\\

The system (\ref{eq:MHDsym}) is hyperbolic,  its Jacobian has real eigenvalues and a complete
set of eigenvectors. However, it is not a strictly hyperbolic system since
some eigenvalues may coincide. Apart from waves moving with the material velocity, 
it is usual to split the set of MHD eigenvalues and associated waves into three groups, 
that are defined as : 

\begin{displaymath}\label{eq:MHDEigenvalues}
 \refstepcounter{equation}
\begin{array}{lll}
 \text{Fast Magnetosonic waves :} & &\\
 \lambda_F^\pm = \vc{u}.\vc{n} \pm  C_F & \text{with~~} 
 C_F^2 = \dfrac{1}{2}(V_t^2 + v_A^2 + \sqrt{(V_t^2 + v_A^2)^2-4 V_t^2 C_A^2}~)& 
 ~~~(\theequation.1)\\
 \\
 \text{Alfen waves :} & &\\ 
 \lambda_A^\pm = \vc{u}.\vc{n} \pm  C_A & \text{with~~}
 C_A^2 = (\vc{B}.\vc{n})^2/\rho&~~~(\theequation.2)\\
 \\
 \text{Slow Magnetosonic waves :} & &\\
 \lambda_S^\pm = \vc{u}.\vc{n} \pm  C_S & \text{with~~}
 C_S^2 = \dfrac{1}{2}(V_t^2 + v_A^2 - \sqrt{(V_t^2 + v_A^2)^2-4 V_t^2 C_A^2})&
 ~~~(\theequation.3)
 \end{array}
\end{displaymath}
where $v_A$ and $V_t$ are defined by : 
$ v_A^2 = |\vc{B}|^2/\rho$ and $V_t^2=\gamma p/\rho $.\\

The velocity of these waves is ordered as follows :
$$
\lambda_S^2 \le  \lambda_A^2 \le \lambda_F^2
$$
Fast and slow Magnetosonic waves are the equivalent of acoustic waves in fluid dynamics. 
Alfen waves (sometimes also called shear Alfen waves) are of a different nature : The expression 
(\ref{eq:MHDEigenvalues}.2) shows that they do not propagate in the direction orthogonal to the 
the magnetic field. Actually in the direction orthogonal to the magnetic field, the speed of 
propagation of Alfen and slow magnetosonic waves is zero (in a frame moving with the material velocity) 
and only the fast magnetosonic waves survive. 


\subsection{Large aspect ratio theory}
\subsubsection{Geometry and coordinate system}
In this section, we are concerned with the model of the ``straight tokamak'' that consists of a 
slender torus characterized by a small aspect ratio $\varepsilon = a/R_0$ (see figure 
\ref{fig:StraighTokamak}). In this model, the torus is approximated by 
a periodic cylinder of length $2 \pi R_0$ and of section of radius $a$. Some of the 
dynamical effects that occur in a tokamak are well represented in this way and this model 
have been extensively used in theoretical studies to understand tokamak dynamics. 
In particular, it is the model considered in \cite{Strauss76} to derive his original 
reduced model.
\begin{center}
 \begin{figure}[h]
 \begin{center}
 \includegraphics[width=0.75\textwidth]{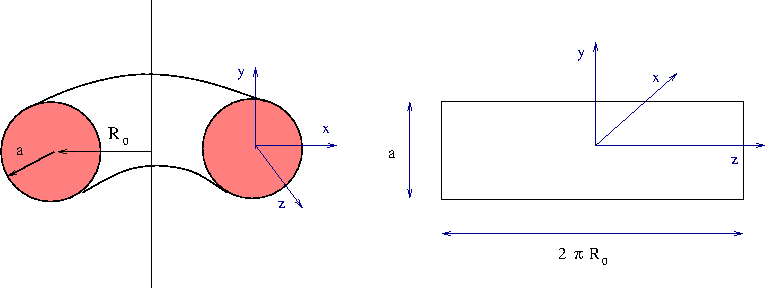}
\caption{Straight tokamak model : the slender torus is unfold to form a periodic cylinder
\label{fig:StraighTokamak}}
\end{center}
 \end{figure}
\end{center}

Now, let $(\xi, \eta, \zeta)$ be the usual cartesian coordinate system and 
let us introduce the normalized variables : 
\begin{displaymath}
\left\lbrace
\begin{array}{l}
x=\xi/a\\
y = \eta/a \\
z=\zeta/R_0
\end{array}
\right.
\end{displaymath} 
In a way consistent with the underlying physical problem, the $z$ direction will be 
denoted as the toroidal direction while the planes $(x,y)$ are the poloidal sections. 
Note also that the introduction of the normalized coordinates $(x,y,z)$ 
corresponds actually to a {\em two scale} analysis : $z$ the toroidal 
coordinate is scaled with $R_0$ while the poloidal coordinates $(x,y)$ are scaled with the 
small radius $a$. \\
With these normalized coordinates, the expression of the spatial operators 
becomes : 
\begin{displaymath}\label{eq:OperatorExpression}
 \refstepcounter{equation}
\begin{array}{llr}
a \nabla f &= \dudx{f}{x}\vecx + \dudx{f}{y}\vecy  + 
{\varepsilon}\dudx{f}{z}\vecz &(\theequation.1)\\
\\
a \nabla \bullet \vc{v} &= \nabla_\perp \bullet \vc{v}_\perp + 
\varepsilon \dudx{\vc{v}_z}{z}&(\theequation.2)\\
\\
a \nabla \times \vc{v} 
&= (\vecz\bullet\nabla_\perp \times \vc{v}_\perp) \vecz + 
\nabla_\perp \vc{v}_z\times \vecz +
\varepsilon (-\vecx\dudx{\vc{v}_y}{z} +\dudx{\vc{v}_x}{z}\vecy)\\
&\\
&= \partial_y \vc{v}_z \vecx - \partial_x \vc{v}_z\vecy
+(\partial_x \vc{v}_y - \partial_y \vc{v}_x)\vecz + \varepsilon(-\partial_z{\vc{v}_y}\vecx
+\partial_z{\vc{v}_x}\vecy)&(\theequation.3)\\
\end{array}
\end{displaymath}
with the definitions : 
\begin{displaymath}
\begin{array}{ll}
\vc{v}=\vc{v}_\perp + \vc{v}_z \vecz & \vc{v}_\perp = \vc{v}_x \vecx+\vc{v}_y \vecy\\
\nabla_\perp f = \dudx{f}{x}\vecx+\dudx{f}{y}\vecy&
\nabla_\perp \bullet \vc{v}_\perp = \dudx{\vc{v}_x}{x}+\dudx{\vc{v}_y}{y}
\end{array}
\end{displaymath}
\subsubsection{Scaling}
We now proceed to scale the unknown variables. To recast the equations into an useful form, the usual procedure is to write 
them in dimensionless form by scaling every
variable by a characteristic value. Here, in an equivalent manner, we will consider the
following change of variables:
 \begin{displaymath}
\label{eq:ScalingAssumption}
 \refstepcounter{equation}
 \begin{array}{lll}
 \text{Magnetic field :} &
 \vc{B}= \dfrac{F}{R}\vecz + \vc{B}_P= {B_0}(\vecz 
 + \varepsilon {\mathcal{B}})& ~~(\theequation.1)\\
 \\
 \text{Pressure :} &
 p = P_0 (\bar{p} + \varepsilon q) & ~~(\theequation.2)\\
  \\
 \text{Velocity :} &
 \vc{u} = \varepsilon v_A \vc{v}& ~~(\theequation.3)\\
  \\
 \text{Time :} &
 t  = \dfrac{a}{\varepsilon v_A} \tau& ~~(\theequation.4)\\
 \end{array}
 \end{displaymath}
 where $B_0$ is the reference value of the toroidal magnetic field on the magnetic axis ($R=R_0$) 
 and $\bar{p}$ is a constant. In these expressions, $v_A$ is the Alfen 
 speed defined by $v_A^2 = B_0^2/\rho_0$ 
 where $\rho_0$ is some reference density (for instance, a characteristic value of 
 the density on the magnetic axis). We choose for simplicity $P_0 = \rho_0 v_A^2$ 
 (this only affects the value of the constant $\bar{p}$). 
 The important assumptions made in (\ref{eq:ScalingAssumption}) are : \\
 \\
 i) that the toroidal magnetic field dominates the flow and that the poloidal field is of order 
 $\varepsilon$ with respect to the toroidal field : 
 $\vc{B}= \vc{B}_T + \varepsilon\vc{B}_P$. 
 In tokamaks, the toroidal field is mainly due to external coils and it varies typically
 as 
 $\vc{B}_T= \dfrac{F}{R}\vecz$ where $F$ is approximately a constant and $R$ is the distance to 
 the rotation axis of the torus. In the model of the ``straight tokamak'' and in the limit of 
 small aspect ratio $a/R_0$, this leads to the following expansion of the magnetic field : 
 $$\vc{B}= \dfrac{F}{R}\vecz + \varepsilon\vc{B}_P = \dfrac{F_0}{R_0(1+\varepsilon x)}\vecz + 
 \dfrac{F-F_0}{R_0(1+\varepsilon x)}\vecz 
 + \varepsilon\vc{B}_P = B_0 (\vecz + \varepsilon {\mathcal{B}}) $$
 where $B_0 = F_0/R_0$ is the value of the toroidal magnetic field on the magnetic axis. 
Note that ${\mathcal{B}}$ contains a toroidal component. This component is assumed 
 to be of the same order than the magnetic poloidal field. \\
 \\
 ii) that the pressure fluctuations are also of order $\varepsilon$. Since the 
 poloidal magnetic field is of order $\varepsilon$ with respect to the toroidal 
 field, this means that the poloidal plasma $\beta$ parameter is of order 1. 
 In the physical literature, this situation is referred to as a ``high'' $\beta$ ordering 
 \cite{Strauss77}.\\\\
 iii) that the velocities are small with respect to the Alfen speed. 
 Strictly speaking this assumption needs only to be done for the perpendicular 
 velocity. We adopt it for the full velocity vector in order to simplify the 
 presentation. \\\\
 iv) that we are interested in the long time behavior. Actually, the assumption 
 (\ref{eq:ScalingAssumption}) means 
 that we are interested in the long time behavior of the system with respect to 
 the Alfen time $a/v_A$ that represent the typical time for a magnetosonic wave to cross the 
 tokamak section. (see section \ref{sec:fastlimit} for some remarks on the short time 
 behavior of the system on the fast scale $a/v_A$). \\
 
 Introducing these expression into the MHD system, we get : 
 \begin{displaymath}\label{scaledMHD}
 \refstepcounter{equation}
\begin{array}{ll}
%
\colorbox{grispale}{$\rho (\bar{p} + \varepsilon q)$}[\dudx{}{\tau}\vc{v} + (\vc{v}_\perp\cdot\nabla_\perp) \vc{v} ]
+ \partial_z(q+{\mathcal{B}}_z)\vecz + \nabla_\perp{\mathcal{B}}^2/2 
 -\partial_z {\mathcal{B}} - ({\mathcal{B}}_\perp\cdot\nabla_\perp){\mathcal{B}}\\
+ \varepsilon (\rho v_z\partial_z \vc{v} + \partial_z({\mathcal{B}}^2/2) \vecz
-{\mathcal{B}}_z\partial_z {\mathcal{B}}) +
\colorbox{grispale}{$\dfrac{1}{\varepsilon}\nabla_\perp (q+{\mathcal{B}}_z)$}
=0 & (\theequation.2)\\\\
\dudx{}{\tau}{{\mathcal{B}}_\perp} + (\vc{v}_\perp\cdot\nabla_\perp){\mathcal{B}}_\perp -
({\mathcal{B}}_\perp\cdot\nabla_\perp)\vc{v}_\perp + 
{\mathcal{B}}_\perp \nabla_\perp\cdot \vc{v}_\perp - 
\partial_z \vc{v}_\perp \\ + 
\varepsilon (\vc{v}_z\partial_z {\mathcal{B}}_\perp-{\mathcal{B}}_z\partial_z \vc{v}_\perp 
+\partial_z \vc{v}_z {\mathcal{B}}_\perp )= 0
& (\theequation.3)\\
\\
\dudx{}{\tau}{{\mathcal{B}}_z} + (\vc{v}_\perp\cdot\nabla_\perp){\mathcal{B}}_z
-
({\mathcal{B}}_\perp\cdot\nabla_\perp)\vc{v}_z + 
{\mathcal{B}}_z \nabla_\perp\cdot \vc{v}_\perp +
\varepsilon \vc{v}_z\partial_z {\mathcal{B}}_z
+ \colorbox{grispale}{$\dfrac{1}{\varepsilon}
\nabla_\perp\cdot \vc{v}_\perp$ }= 0 & (\theequation.4)
\\\\
\colorbox{grispale}{$\dfrac{1}{\gamma(\bar{p} + \varepsilon q)}$}
[\dudx{}{\tau}{q} +(\vc{v}_\perp.\nabla_\perp) q +  
{\varepsilon} \vc{v}_z \partial_z {q}]
+ \partial_z \vc{v}_z 
+ \colorbox{grispale}{$\dfrac{1}{\varepsilon} \nabla_{\!\!\perp} \cdot
\vc{v}_\perp $}= 0& (\theequation.5)\\
\end{array}
\end{displaymath} 

\newpage
If one introduces the variable $\vc{W}=(\vc{v}_x,\vc{v}_y,\vc{v}_z,
{\mathcal{B}}_x,{\mathcal{B}}_y,{\mathcal{B}}_z,q)^t$, the previous system can be written as 
\begin{equation}\label{eq:matricialSymetricform}
A_0(\varepsilon \vc{W}) \partial_\tau \vc{W} + \sum_j A_j (\vc{W}, \varepsilon \vc{W})\partial_{x_j} \vc{W}+ \frac{1}{\varepsilon} 
\sum_j C_j\partial_{x_j} \vc{W} =0 
\end{equation} 
where the matrices $A_0,A_j(\vc{W},\varepsilon \vc{W})$ are defined by : 
\begin{footnotesize}
\begin{displaymath}A_0=
\left(\begin{array}{ccc}
 \rho I_3& 0_3 & 0 \\
  0_3     & I_3 & 0 \\
  0_3     & 0_3 & \dfrac{1}{\gamma(\bar{p} + \varepsilon q)}\\ 
\end{array}\right)
\end{displaymath}
\begin{displaymath}A_x=
\left(\begin{array}{ccccccc}
  \rho \vc{v}_x & 0 & 0& 0 &  {\mathcal{B}}_y& {\mathcal{B}}_z & 0 \\
  0 & \rho \vc{v}_x & 0 & 0 & -{\mathcal{B}}_x& 0               & 0 \\
  0 & 0 & \rho \vc{v}_x & 0 & 0               & -{\mathcal{B}}_x& 0 \\
  0 & 0 & 0 & \vc{v}_x & 0 & 0 & 0 \\
  {\mathcal{B}}_y & -{\mathcal{B}}_x & 0& 0 & \vc{v}_x& 0 & 0 \\
  {\mathcal{B}}_z & 0 & -{\mathcal{B}}_x & 0 & 0 & \vc{v}_x & 0 \\
  0 & 0 & 0 & 0 & 0 & 0 & \dfrac{\vc{v}_x}{\gamma(\bar{p} + \varepsilon q)}\\ 
\end{array}\right)
\end{displaymath}
\begin{displaymath}
A_y=
\left(\begin{array}{ccccccc}
 \rho \vc{v}_y & 0 & 0 & -{\mathcal{B}}_y& 0 & 0               & 0 \\
 0 & \rho \vc{v}_y& 0 & {\mathcal{B}}_x & 0 & {\mathcal{B}}_z & 0 \\
 0 & 0 & \rho \vc{v}_y& 0               & 0 & -{\mathcal{B}}_y& 0 \\
 -{\mathcal{B}}_y & {\mathcal{B}}_x& 0& \vc{v}_y& 0 & 0 & 0\\
 0 & 0 & 0 & 0 & \vc{v}_y & 0 & 0 \\
 0 &  {\mathcal{B}}_z& -{\mathcal{B}}_y& 0 & 0 & \vc{v}_y & 0\\
 0 & 0 & 0 & 0 & 0 & 0 & \dfrac{\vc{v}_y}{\gamma(\bar{p} + \varepsilon q)}\\ 
\end{array}\right)
\end{displaymath}
\begin{displaymath}A_z=
\left(\begin{array}{ccccccc}
  \varepsilon \rho \vc{v}_z & 0 & 0 & -{(1+\varepsilon {\mathcal{B}}_z)} & 0  & 0& 0 \\
  0 & \varepsilon \rho \vc{v}_z & 0 & 0 & -{(1+\varepsilon {\mathcal{B}}_z)}  & 0& 0 \\
  0 & 0 & \varepsilon \rho \vc{v}_z & {\varepsilon {\mathcal{B}}_x}& {\varepsilon {\mathcal{B}}_y}& 0& 1 \\
  - {(1+\varepsilon {\mathcal{B}}_z)} & 0 
 & {\varepsilon{\mathcal{B}}_x}& 
 \varepsilon\vc{v}_z & 0 & 0 & 0 \\
 0 & -({1+\varepsilon{\mathcal{B}}_z}) & {\varepsilon{\mathcal{B}}_y}
 & 0 &\varepsilon\vc{v}_z  & 0 & 0 \\
  0 & 0 & 0 & 0 & 0 & \varepsilon\vc{v}_z & 0 \\
 0 & 0 & 1 & 0 & 0 & 0 & \varepsilon \dfrac{\vc{v}_z}{\gamma(\bar{p} + \varepsilon q)}\\  
\end{array}\right)
\end{displaymath}
\end{footnotesize}
while the constant matrices $C_j$ are given by :
\begin{footnotesize}
\begin{displaymath}
C_x= 
\left(\begin{array}{ccccccc}
0 &0 &0 &0 &0 & 1 & 1 \\
0 &0 &0 &0 &0 &0 & 0 \\
0 &0 &0 &0 &0 &0 & 0 \\
0 &0 &0 &0 &0 &0 & 0 \\
0 &0 &0 &0 &0 &0 & 0 \\
1 & 0 &0 &0 &0 &0 & 0 \\ 
1 & 0 &0 &0 &0 &0 & 0 \\ 
\end{array}\right)
C_y= 
\left(\begin{array}{cccccccc}
0 &0 &0 &0 &0 & 0 & 0 \\
0 &0 &0 &0 &0 & 1 & 1 \\
0 &0 &0 &0 &0 &0 & 0 \\
0 &0 &0 &0 &0 &0 & 0 \\
0 &0 &0 &0 &0 &0 & 0 \\
0 & 1 & 0 &0 &0 &0 & 0 \\ 
0 & 1 & 0 &0 &0 &0 & 0 \\ 
\end{array}\right)
\end{displaymath}
\end{footnotesize}
This form makes apparent that the ideal MHD system can be put under the general form studied in 
section \ref{general framework}. Therefore, the general results obtained in this section 
can be applied and we have 
\newtheorem{RMHD}{Theorem 2}
\begin{RMHD}
Assume that the initial velocity, magnetic field and pressure 
are defined by : 
\begin{displaymath}
\left\lbrace\begin{array}{l}
\vc{u}(0,\vc{x})/V_A=\varepsilon(\vc{v}^0(\vc{x})+ \varepsilon \vc{v}^1(\varepsilon,\vc{x}))\\
\vc{B}(0,\vc{x})/B_0= \vc{e}_z + \varepsilon({\mathcal{B}}^0(\vc{x})
+ \varepsilon {\mathcal{B}}^1(\varepsilon,\vc{x}))\\
p(0,\vc{x})/p_0= \bar{p} + \varepsilon(q^0(\vc{x})+ \varepsilon q^1(\varepsilon,\vc{x}))\\
\end{array}\right.
\end{displaymath}
where $\bar{p}$ is a constant, the functions $\vc{v}^0, {\mathcal{B}}^0,q^0$ and 
$ \vc{v}^1, {\mathcal{B}}^1,q^1$ are bounded in $H^s$ and where the $0$-th order initial data 
verifies : 
\begin{displaymath} \label{eq:Theo2Initialdata}
\refstepcounter{equation}
\left\lbrace\begin{array}{ll}
\nabla_\perp .\vc{v}^0(\vc{x})=0 &~~~~~~~~~~(\theequation.1)\\
\exists f(z) \mathrm{~such~that~} {\mathcal{B}}_z^0(\vc{x})=f(z)-q^0(\vc{x})
&~~~~~~~~~~(\theequation.2)
\end{array}\right.
\end{displaymath}
then the solution of the full MHD system (\ref{eq:MHDsym}) 
exists for a time $T$ independent of
$\varepsilon$ and this solution converges in $H^{s-1}$ to the solution of the reduced system 
given below in section \ref{sec:slowlimit}. \\
\end{RMHD}

\begin{proof} The conditions on the structure of the system given in theorem 1
are satisfied while the conditions (\ref{eq:Theo2Initialdata}) express the fact that 
the $0-$order initial data is in the kernel of the large operator. The assumptions of theorem
1 are then fulfilled and the result follows.
\end{proof}

\subsubsection{Slow limit of the system}\label{sec:slowlimit}
According to the general theory described in section \ref{general framework}, the solutions of 
(\ref{eq:MHDsym}) will be close to the solutions of the limit system of equations given by
\begin{equation}\label{eq:limitsystem}
\left\{\begin{array}{l}
A_0(0) \partial_\tau \vc{W^0} + A_j (\vc{W}^0, 0)\partial_{x_j} \vc{W^0}+  
C_j\partial_{x_j} \vc{W^1} =0  \\
C_j\partial_{x_j} \vc{W^0}=0
\end{array}\right.
\end{equation} 
The zero-order solutions are 
functions $\vc{W}^0=(\vc{v}_x,\vc{v}_y,\vc{v}_z,
{\mathcal{B}}_x,{\mathcal{B}}_y,{\mathcal{B}}_z,q)$ that are in the kernel of 
the large operator. These functions must therefore verify :
\begin{displaymath}\label{kernel}
 \refstepcounter{equation}
\begin{array}{ll}
 \nabla_\perp \cdot\vc{v}_\perp = 0 &  ~~(\theequation.1)\\
 \nabla_\perp (q + {\mathcal{B}}_z) = 0 &  ~~(\theequation.2)\\
 \end{array}
\end{displaymath}
using these results, the explicit form of system (\ref{eq:limitsystem}) can be written 
\begin{displaymath}\label{scaledMHDlimit}
 \refstepcounter{equation}
\begin{array}{ll}
\rho (\bar{p})[\dudx{}{\tau}\vc{v}_z + (\vc{v}_\perp\cdot\nabla_\perp) \vc{v}_z ]
+ \partial_z q + ({\mathcal{B}}_\perp\cdot\nabla_\perp) q 
=0 & (\theequation.1)\\\\
\dudx{}{\tau}{{\mathcal{B}}_\perp} + (\vc{v}_\perp\cdot\nabla_\perp){\mathcal{B}}_\perp -
({\mathcal{B}}_\perp\cdot\nabla_\perp)\vc{v}_\perp - 
\partial_z \vc{v}_\perp = 0
& (\theequation.2)\\
\\
\rho (\bar{p})[\dudx{}{\tau}\vc{v}_\perp + (\vc{v}_\perp\cdot\nabla_\perp) \vc{v}_\perp ]
+ \nabla_\perp{\mathcal{B}}^2/2 
 -\partial_z {\mathcal{B}}_\perp - ({\mathcal{B}}_\perp\cdot\nabla_\perp){\mathcal{B}}_\perp\\
+ \colorbox{grispale}{$\nabla_\perp (q^1+{\mathcal{B}}^1_z)$}
=0 & (\theequation.3)\\\\
\dudx{}{\tau}{{\mathcal{B}}_z} + (\vc{v}_\perp\cdot\nabla_\perp){\mathcal{B}}_z
-
({\mathcal{B}}_\perp\cdot\nabla_\perp)\vc{v}_z  
+ \colorbox{grispale}{$
\nabla_\perp\cdot \vc{v}^1_\perp$ }= 0 & (\theequation.4)
\\\\
\dfrac{1}{\gamma\bar{p}}
[\dudx{}{\tau}{q} +(\vc{v}_\perp.\nabla_\perp) q ]
+ \partial_z \vc{v}_z 
+ \colorbox{grispale}{$\nabla_{\!\!\perp} \cdot\vc{v}^1_\perp $}= 0& (\theequation.5)\\
\end{array}
\end{displaymath} 
where $\bar{p}$ and $\rho (\bar{p})$ are constant. Using the fact that by (\ref{kernel}.2) 
$q + {\mathcal{B}}_z = f(z)$ where $f(z)$ is an arbitrary function, equations (\ref{scaledMHDlimit}.4) and 
(\ref{scaledMHDlimit}.5) can be combined to eliminate the corrective term 
$\nabla_{\!\!\perp} \cdot\vc{v}^1_\perp $ resulting in : 
\begin{displaymath}\label{eq:pressureEqFinal}
(\dfrac{1}{\gamma\bar{p}}-1)
[\dudx{}{\tau}{q} +(\vc{v}_\perp.\nabla_\perp) q ]
+ ({\mathcal{B}}_\perp\cdot\nabla_\perp)\vc{v}_z + \partial_z \vc{v}_z= 0
\end{displaymath}
We also note that in the perpendicular momentum equation, the term $\nabla_\perp (q^1+{\mathcal{B}}^1_z)$ 
ensures that $\nabla_\perp \cdot\vc{v}_\perp = 0$, 
this term can therefore be combined with the $\nabla_\perp{\mathcal{B}}^2/2$ term with no 
change in the result. Introducing the notations 
\begin{displaymath}
 \DtPerp{\cdot} = \dudx{}{\tau}{\cdot} + (\vc{v}_\perp\cdot\nabla_\perp)\cdot
 ~~~~
  \gradPar{\cdot} = 
   ({\mathcal{B}}_\perp\cdot\nabla_\perp)\cdot   +  \partial_z \cdot
\end{displaymath}

we get the final limit system : 
\begin{displaymath}\label{scaledMHDlimitFinal}
 \refstepcounter{equation}
\begin{array}{ll}
\rho \DtPerp{\vc{v}_\perp}
- \gradPar{{\mathcal{B}}_\perp} 
+ \colorbox{grispale}{$\nabla_\perp \lambda$}
=0 & (\theequation.1)\\\\
\DtPerp{{\mathcal{B}}_\perp}
- \gradPar{\vc{v}_\perp}  = 0
& (\theequation.2)\\
\\
\rho \DtPerp{\vc{v}_z} + \gradPar{q} 
=0 & (\theequation.3)\\\\
(\dfrac{1}{\gamma\bar{p}}-1)
\DtPerp{q}
+\gradPar{\vc{v}_z} = 0
& (\theequation.4)\\\\
\end{array}
\end{displaymath} 
where $\lambda$ stands here for a scalar ``pressure'' that ensures that the perpendicular 
divergence of the perpendicular velocity is zero. \\

The equations (\ref{scaledMHDlimitFinal}) shows that the limit system splits into two different 
sub-systems : 
\begin{itemize}
 \item (\ref{scaledMHDlimitFinal}.1) and (\ref{scaledMHDlimitFinal}.2) as well as the 
 constraint 
 (\ref{kernel}.1) describe the 
{\em incompressible} dynamics of the perpendicular motion of the plasma. 
This set of equation does 
not depend on the pressure and toroidal velocity equations and can be solved independently
of the other two equations. 
 \item On the other hand, the two scalar equations (\ref{scaledMHDlimitFinal}.3) and 
 (\ref{scaledMHDlimitFinal}.4) 
 describe the {\em compressible} parallel dynamics of the plasma. Actually, without the perpendicular 
 convective terms, these two equations describe a compressible one dimensional flow in the parallel 
 direction to the magnetic field. Note that these equations are 
 ``slave'' of the first two ones since both the perpendicular advection and the $\nabla_{/\!/}$ operator 
 depend only on the solution of equations (\ref{scaledMHDlimitFinal}.1) and (\ref{scaledMHDlimitFinal}.2). 
 Thus (\ref{scaledMHDlimitFinal}.3) and 
 (\ref{scaledMHDlimitFinal}.4) can be solved once the solutions of 
 (\ref{scaledMHDlimitFinal}.1) and (\ref{scaledMHDlimitFinal}.2) have been computed. 
\end{itemize}

As in the original MHD system, the system 
(\ref{scaledMHDlimitFinal}) is endowed with an involution : Using  
that $\nabla.\vc{B}=0$,  we have in the limit $\varepsilon \rightarrow 0$
that the perpendicular divergence of the magnetic field is zero,
\begin{displaymath}\label{eq:perpDiv}
\nabla_\perp \cdot {\mathcal{B}}_\perp = 0
\end{displaymath}
if this property is true for the initial data, it is conserved by system 
(\ref{scaledMHDlimitFinal}) :  
\newtheorem{PerpDiv}{Proposition}
\begin{PerpDiv}
 Assume that the perpendicular divergence of the perpendicular magnetic field is zero at time 
 $t=0$ : 
 $~~\nabla_\perp .{\mathcal{B}}_\perp(\vc{x},t=0) = 0 $ 
then $\nabla_\perp .{\mathcal{B}}_\perp(\vc{x},t) = 0$ for $t>0$.
\end{PerpDiv}
\begin{proof}
 This follows directly by applying the perpendicular divergence operator to 
 the perpendicular Faraday law (\ref{scaledMHDlimitFinal}.2). Note that to obtain this result, 
 both the properties $\nabla_\perp .{\mathcal{B}}_\perp =0$ and $\nabla_\perp .\vc{v}_\perp =0$
 are used. 
\end{proof}

Although, equations (\ref{scaledMHDlimitFinal}.1) and (\ref{scaledMHDlimitFinal}.2) have 
a similar structure, we note that $\nabla_\perp \cdot\vc{v}_\perp = 0$ {\em is not} an 
involution for the system : equation (\ref{scaledMHDlimitFinal}.1) does not conserve the 
perpendicular divergence of $\vc{v}_\perp$, the corrective term $\nabla_\perp \lambda$ is 
therefore needed to insure that $\nabla_\perp \cdot\vc{v}_\perp = 0$. \\

We will now from the limit system (\ref{scaledMHDlimitFinal}) obtain a {\em reduced} system 
characterized by a smaller number of equation than the number of the original system. As explained in 
section \ref{sec:reduced}, this can be obtained by canceling out the corrective term. Since equations 
(\ref{scaledMHDlimitFinal}.1) and (\ref{scaledMHDlimitFinal}.2) form an autonomous system, we 
concentrate on these two equations. According to the general procedure sketched in 
section \ref{sec:reduced}, we look for a parametrization of the function space where the 
solution belongs to. In the present case, the space 
$K=\{ (\vc{v}_\perp,{\mathcal{B}}_\perp); 
\nabla_\perp \cdot\vc{v}_\perp =\nabla_\perp\cdot{\mathcal{B}}_\perp = 0\} $ can be 
parametrized  by 2 scalar functions $\phi, \psi$ such that 
\begin{displaymath}\label{eq:parametrization}
 \refstepcounter{equation}
\begin{array}{ll}
 \vc{v}_\perp = \vecz \times \nabla \phi &  ~~(\theequation.1)\\
 {\mathcal{B}}_\perp = \vecz \times \nabla \psi &  ~~(\theequation.2)
 \end{array}
\end{displaymath}
Let us define for any scalar function $F \in H^1$ the operator ${\cal{M}}$ with values in 
$L^2 \times L^2$ by :
$${\cal{M}}(F)= \vecz \times \nabla F 
$$
The following Green formula :
$$\int_\Omega \vecz \times \nabla F \cdot \vc{W} d\vc{x} = \int_{\partial\Omega} F \vecz \times \vc{W} \cdot 
\vc{n} ds - \int_\Omega F \vecz \cdot \nabla \times \vc{W} d\vc{x}
$$
shows that the adjoint operator of ${\cal{M}}$ is defined by : 
\begin{displaymath}
 {\cal{M}}^*(\vc{W}) =\vecz \cdot\nabla \times\vc{W}
\end{displaymath} 
Using the general recipe given in section \ref{sec:reduced}, we get a reduced system for the variables 
$\phi, \psi$ by :

\begin{displaymath}
\label{eq:reducedModel1} \refstepcounter{equation}
\begin{array}{ll}
\rho {\cal{M}}^*\DtPerp{{\cal{M}}(\phi)}
- {\cal{M}}^*\gradPar{{\cal{M}}(\psi)} 
=0 & (\theequation.1)\\\\
{\cal{M}}^*\DtPerp{{\cal{M}}(\psi)}
- {\cal{M}}^*\gradPar{{\cal{M}}(\phi)}  = 0
& (\theequation.2)\\
\end{array} 
 \end{displaymath}
 where the corrective term $\nabla_\perp \lambda$ have been canceled out by the annhilator operator 
${\cal{M}}^*$. After some algebra, this system admits the following expression :  
\begin{displaymath}\label{eq:MHDlimitReduced2}
 \refstepcounter{equation}
\begin{array}{ll}
\rho \DtPerp{{\cal{U}}}
- \gradPar{J} 
=0 & (\theequation.1)\\\\
\partial_\tau {J}
- \nabla_\perp^2(\partial_x \phi \partial_y \psi -\partial_x \psi \partial_y \phi )
-\partial_z{{\cal{U}}}  = 0
& (\theequation.2)\\
\end{array}
\end{displaymath} 
 where ${\cal{U}}$ and $J$ are defined as ${\cal{U}} = -  \nabla_\perp^2 \phi$ and 
 $J = - \nabla_\perp^2 \psi$. We note that ${\cal{U}}= \partial_y \vc{v}_x - \partial_x \vc{v}_y$
 represent the $z-$component of the curl of the velocity vector, therefore 
 in reduced MHD literature, ${\cal{U}}$ is defined as the {\em vorticity} and 
 (\ref{eq:MHDlimitReduced2}.1) 
 is called the vorticity equation by analogy with the fluid dynamics case.\\
 
 From a physical point of view, the quantity $J = - \nabla_\perp^2 \psi$ corresponds to the 
 toroidal current traversing the plasma column and therefore 
 equation (\ref{eq:MHDlimitReduced2}.2) defines the behavior of the toroidal current. 
 In the framework of reduced MHD model, this equation is not used. 
 Instead, rewriting (\ref{eq:MHDlimitReduced2}.2) as 
 
 \begin{displaymath}\label{eq:MHDlimitReduced3}
 \refstepcounter{equation}
 - \nabla_\perp^2 [\partial_\tau {\psi}
+ (\partial_x \phi \partial_y \psi -\partial_x \psi \partial_y \phi )
-\partial_z \phi ] = 0~~~~~~~~~~~~~~(\theequation)
 \end{displaymath}
 and noting that $(\partial_x \phi \partial_y \psi -\partial_x \psi \partial_y \phi )$ 
 corresponds to the advection term $\vc{v}_\perp\cdot\nabla_\perp \psi$, 
 one prefers to use the equation :
\begin{equation}\label{eq:Faradayf}
\frac{\partial }{\partial \tau} \psi
+ \vc{v}_\perp\cdot\nabla_\perp \psi - \frac{\partial}{\partial z} \phi = 0
\end{equation}
Strictly speaking (\ref{eq:Faradayf}) cannot be deduced directly from (
\ref{eq:MHDlimitReduced3}) and integration factors should have appeared in 
(\ref{eq:Faradayf}). However, it is possible to establish directly (\ref{eq:Faradayf}). 
This is done in Annex 1.\\ 

To complete the description of the reduced MHD models, we mention that in the present model, it is 
not necessary to solve the toroidal and pressure equations 
(\ref{scaledMHDlimitFinal}.3) and (\ref{scaledMHDlimitFinal}.4) since the dynamics is entirely governed 
by (\ref{scaledMHDlimitFinal}.1) and (\ref{scaledMHDlimitFinal}.2). Neglecting these equations, is also 
sometimes justified as follows (\cite{Strauss77}) :  
The acceleration term of the 
toroidal momentum equation (\ref{scaledMHDlimitFinal}.3) is  :
$$
\nabla_{/\!/} q = {\mathcal{B}}_\perp\cdot\nabla_\perp q 
  +{\frac {\partial }{\partial z}}q  =  \vc{B}\cdot \nabla q
$$
Then it can be shown (see \cite{Strauss77}), that if at time $t=0, \vc{B}\cdot \nabla q =0$, then this quantity
will stay equal to zero. Therefore, the toroidal acceleration is null and if initially 
$\vc{v}_z = 0 $, then the toroidal velocity will remain zero. Consequently, the velocity source 
$\nabla_{/\!/} \vc{v}_z$ in the pressure equation remains zero and the pressure correction $q$ 
behaves as a passive scalar. \\ In the framework of MHD 
studies in tokamaks, the assumption $\vc{B}\cdot \nabla q = 0 $ is very natural since the 
flows under investigation are close to an MHD equilibrium characterized by : 
\begin{equation}\label{eq:equilibrium}
\nabla p = \vc{J} \times \vc{B}
\end{equation}
that implies that $\vc{B}\cdot \nabla p = 0 $.  
Actually, a lot of MHD studies aims to examine the linear or non-linear 
stability of such equlibrium and therefore these works 
use precisely the relation (\ref{eq:equilibrium}) to define the initial conditions. 

Summarizing, the dynamics of the MHD model can be reduced to a system 
of 2 equations for the scalar quantities $(\phi, \psi ) $

\begin{displaymath}
\label{eq:reducedModel2}
 \refstepcounter{equation}
 \left[ \begin {array}{cl} 
\rho \DtPerp{{\cal{U}}}
- \gradPar{J} 
=0 & (\theequation.1)\\\\
\noalign{\medskip}
\DtPerp{\psi } - \frac{\partial}{\partial z} \phi = 0
& ~~(\theequation.2)\\\\
  \end {array} \right] 
\end{displaymath}
with 
$$
{\cal{U}} = - \nabla_\perp^2 \phi~~~~~~~~~~~~~~J = - \nabla_\perp^2 \psi
$$
These equations are conventionally written in a somewhat different form emphazing 
their hamiltonian character \cite{Morrison-Hazeltine84}. Introducing the bracket
$$
[f,g] = \vecz\cdot \nabla_\perp f \times \nabla_\perp g 
$$
we have that for any $f$
$$
\vc{v}_\perp\cdot \nabla_\perp f = [\phi,f] ~~~~\mathrm{while} ~~~~
\vc{{\cal{B}}}_\perp\cdot \nabla_\perp f = [\psi,f]$$
and the previous system can be written as 
\begin{displaymath}
\label{eq:reducedModel3}
 \refstepcounter{equation}
 \left[ \begin {array}{ll} 
 \frac{\partial }{\partial \tau} {\cal{U}}
+ [\phi,{\cal{U}}] - [\psi,J] - \frac{\partial}{\partial z} J = 0
& ~~(\theequation.1)\\\\
\noalign{\medskip}
\frac{\partial }{\partial \tau} \psi
+ [\phi,\psi] - \frac{\partial}{\partial z} \phi = 0
& ~~(\theequation.2)\\\\
  \end {array} \right] 
\end{displaymath}
where we have assumed $\rho = 1$ using an appropriate scaling of the density. 
\subsubsection{Fast modes of the system}\label{sec:fastlimit}
In this section, we briefly comment on the solutions of the full MHD system that are eliminated by the 
reduced model. In other term, we analyze the short time behavior of system (\ref{eq:matricialSymetricform}). 
Considering 
the fast time scale $\tau' = \varepsilon \tau$ or in an equivalent manner the fast reference 
time $t'= \dfrac{a}{ v_A}$, 
it is seen that the system (\ref{eq:matricialSymetricform}) reduces to the linear hyperbolic system 
\begin{equation}\label{eq:fasttime scale}
A_0(0) \partial_{\tau'} \vc{W} + \sum_j C_j\partial_{x_j} \vc{W} =0 
\end{equation} 
Let $\vc{n}=(\vc{n}_x,\vc{n}_y)^t$ be a 2D unit vector in the poloidal plane, the matrix 
$A_0(0)^{-1} (\vc{n}_x C_x + \vc{n}_y C_y)$ is diagonalizable and its eigenvalues are : 
\begin{equation}\label{eq:eigenvalue}
 \lambda_0 = 0 \text{~(with~multiplicity 6)},~
 \lambda_+ = \sqrt{\frac{\gamma \bar{p} + 1}{\rho}},~
 \lambda_- = -\sqrt{\frac{\gamma \bar{p} + 1}{\rho}}
\end{equation}
or in term of non-normalized variables : 
\begin{equation}\label{eq:eigenvaluephysical}
 \lambda_0 = 0 \text{~(with~multiplicity 6)},~
 \lambda_+ = \sqrt{\frac{\gamma p + B_0^2}{\rho}},~
 \lambda_- = -\sqrt{\frac{\gamma p + B_0^2}{\rho}}
\end{equation}
Comparing these expression to (\ref{eq:MHDEigenvalues}), it is readily be seen that the non-zero 
eigensolutions correspond to fast magnetosonic waves traveling in the direction 
perpendicular to the toroidal magnetic field $B_0 \vecz$. The situation here is quite similar to 
the one encountered with the compressible Euler equation where the fast limit 
corresponds to the acoustic equations describing the propagation of acoustic waves. 
Here, however, 
we also have an additional splitting in term of space directions. The fast limit of the system 
describes the propagation of magnetosonic waves {\em{in the poloidal plane}} while waves traveling in the toroidal 
direction are not present in this limit. The slow limit of the system that have been examined in 
section \ref{sec:slowlimit} thus 
excludes perpendicular magnetosonic waves in the same way as acoustic waves are filtered out from 
the compressible Euler equation when one consider the incompressible limit equation. 
\begin{figure}[h]
\begin{center}
 \includegraphics[width=0.75\textwidth]{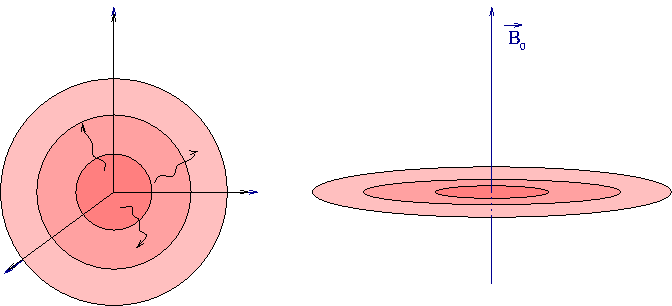}
\caption{Comparison of the fast modes between the low Mach number limit and reduced MHD model; 
Left, Low Mach number limit : 3D isentropic propagation of acoustic waves; 
Right, reduced MHD models : 
2D propagation of fast magnetosonic waves in the poloidal plane.\label{fig:fastmodes}}
\end{center}
\end{figure}

From a numerical point of view, this is one of the main advantage of reduced MHD 
since the use of the full MHD system (\ref{eq:MHDsym}) implies strong CFL stability requirement 
linked to the propagation of magnetosonic waves. Note however that this splitting of the waves is not 
due to differences in the {\em speed of propagation} as for the Euler equation but 
rather to the different {\em space scales} in the toroidal and poloidal directions.
In a toroidal system as a tokamak is, gradients in the toroidal directions are small 
with respect to gradient in the perpendicular directions and it is this fact that 
produce the wave separation rather than their speed of propagation since the velocities of 
Alfen and magnetosonic waves are roughly of the same order of magnitude. 
\section{Concluding remarks}
This work has shown that the derivation of reduced MHD models for fusion plasma can be 
formulated in the general framework of the singular limit of hyperbolic system of PDEs
with large operator. This allows to use the results of this theory and to prove 
rigorously the validity of these approximations. In particular, 
it is proven, that the solutions of the full MHD system converge to the solutions of the 
reduced model displayed in section \ref{sec:slowlimit}.\\
 
\noindent This work can be extended in several different directions.\\

First, the reduced MHD model considered in this paper is the simplest of a whole 
hierarchy of models of increasing complexity. The model used in the present work 
has at least two important weaknesses : \\
a) It uses the straight tokamak model and 
therefore curvature terms are absent from the resulting equations. 
More elaborated models \cite{Carreras-1981,Izzo-etal-1983} retaining curvature effects and high order terms in 
$\varepsilon$ are available and can be possibly 
analyzed within the present framework.\\
b) Another weakness of the model is that it uses as small parameter the ratio $a/R_0$ 
that cannot be considered as small in a large number of today's machines. 
More elaborated models denoted in several references as ``generalized reduced MHD models'' 
\cite{Hazeltine19851,Kruger_etal98,Zeiler_etal1997,Simakov_Catto2004} have been derived. These models do not make use 
of the small aspect ratio hypothesis and thus are in principle applicable with no restriction 
on the geometry. 
However, even from the point of 
view of formal asymptotics, these models are not always easy to understand and 
contains ad-hoc assumptions that are difficult to justify rigorously.  
It would be extremely valuable to study the possibility to formulate these ``generalized 
reduced'' MHD models along the lines exposed in this work. \\

In the terminology of  \cite{S_Schochet_1994}, the present work has examined the 
slow singular limit of the MHD equations. A second possible and interesting sequel of 
this work would be to examine the fast singular limit where no assumption is made on the 
boundedness of the initial time derivatives. On physical grounds, the assumption underlying the 
use of reduced MHD models is that fast transverse magnetosonic waves do not affect the dynamics 
on the long time scale in the same way as in fluid dynamics, the propagation of acoustic waves 
do not modify the average incompressible background. For the Euler (or Navier-Stokes) equations 
this can be proven for certain cases e.g \cite{S_Schochet_1994,Alazard_2008}. Such a result 
however appears significantly more difficult to obtain for the MHD equations since their degree of 
non-linearity is higher than in the fluid dynamics case. Note however, that the formal asymptotic 
expansion used in \cite{Kruger_etal98} can be considered as a first step in this direction. 
\clearpage
\begin{center}
 {\bf Acknowledgments}
\end{center}
This work has benefited from numerous discussion with Philippe Gendhrih and 
Patrick Tamain of IRFM-CEA on the drift approximation in plasma physics. 
Many thanks also to Guido Huysmans of ITER.org for answering my 
(too many and too naive) 
questions on reduced MHD models and finally a special appreciation for Paolo Ricci of 
EPFL-Lausanne for his seminar of April 2014 at the university of Nice and the long 
discussion that follows that helps me a lot to understand reduced MHD models. \\

This work has been partly carried out within the framework of the EUROfusion Consortium and 
has received funding from the Euratom research and training programme 2014-2018 under 
grant agreement No 633053. The views and opinions expressed herein do not necessarily 
reflect those of the European Commission.
\clearpage
\appendix
\section{Annex 1}
In this section, we give a direct obtention of Equation (\ref{eq:Faradayf}). 
Since $\vc{B}$ is a divergence free vector field, there exists a vector potential $\vc{A}$  such that $\nabla \times \vc{A} = \vc{B}$. 
From the expression (\ref{eq:OperatorExpression}.3) of the curl operator and the expression of the 
magnetic field, it is seen that $\psi$ is 
the toroidal component of this vector potential. 
In term of vector potential $\vc{A}$, Faraday's law writes : 
\begin{displaymath}
\frac{\partial }{\partial t} \vc{A} + \vc{E} =  \nabla \phi 
\end{displaymath}
where $\vc{E}$ is the electric field and $\phi$ is  
the electric potential\footnote{Note that the sign convention to define the electric field can be the opposite depending on the authors}\\
Taking the scalar product of this equation by $\vecz$, one has 
\begin{displaymath}
\frac{\partial }{\partial t} \psi + \vecz\cdot \vc{E} -  \partial_z \phi = 0
\end{displaymath}
Now, using Ohm's law 
$\vc{E} + \vc{v} \times \vc{B} = 0$ and the identity 
$$ - \vecz\cdot (\vc{v} \times \vc{B}) = \vc{v}\cdot(\vecz \times \vc{B} )$$
one obtains (\ref{eq:Faradayf}) : 
\begin{displaymath}
 \frac{\partial }{\partial \tau} \psi
+ \vc{v}_\perp\cdot\nabla_\perp \psi - \frac{\partial}{\partial z} \phi = 0
\end{displaymath}
Note that since $\vc{v}_\perp\cdot\nabla_\perp \psi = - {\mathcal{B}}_\perp\cdot\nabla_\perp \phi$, this equation can also be written 
\begin{displaymath}
 \frac{\partial }{\partial \tau} \psi - \nabla_{/\!/} \phi = 0
\end{displaymath}
From a physical point of view, this interpretation shows that the velocity defined 
by (\ref{eq:parametrization}.1) 
is the so-called electric drift $\vec{v}_E=\vc{E}\times \vc{B}/|B^2|$ 
Indeed it can be shown (see \cite{Strauss77} for instance) that the reduced MHD approximation implies that the transverse electric field is 
electrostatic : $$\vc{E}_\perp = \nabla_\perp \phi$$
from which one can deduce by taking the cross product of Ohm's law by $\vc{B}$ the expression (\ref{eq:parametrization}.2) since in the small aspect ratio theory, 
the parallel and toroidal direction are identical up to terms of order $\varepsilon$. 
\clearpage

\bibliographystyle{siam}
\bibliography{reducedMHD,low_mach}

\end{document}